\documentclass{article}
\usepackage[utf8]{inputenc}
\usepackage[T1]{fontenc}
\usepackage{amsmath}
\usepackage{dsfont}
\usepackage{fancyhdr}
\usepackage[table,xcdraw]{xcolor}
\usepackage{graphicx}
\usepackage{multirow}
\usepackage{endnotes}
\usepackage{float}
\usepackage{vmargin}
\usepackage{amssymb}
\usepackage{subfigure}
\usepackage{import}
\usepackage{stackrel}
\usepackage{wrapfig}
\usepackage{amsthm}
\usepackage{rotating}
\usepackage[title]{appendix}
\usepackage{listings}
\usepackage{soul}
\usepackage{stmaryrd}
\usepackage{tikz}
\usepackage{bbm}
\usepackage{mleftright}
\mleftright
\usetikzlibrary{matrix}
\usepackage{tikz}
\usetikzlibrary{shapes,arrows,chains, decorations.pathmorphing, decorations.pathreplacing}
\tikzset{quantum/.style={decorate, decoration=snake}}

\newcommand{\ket}[1]{\left|#1\right\rangle}		

\DeclareMathOperator{\Tr}{Tr}

\newcommand{\ketbra}[2]{\left|#1\rangle\langle#2\right|}

\newcommand{\abs}[1]{\lvert #1\rvert}

\newcommand{\eps}{\varepsilon}

\newcommand{\QPVBBf}{$\mathrm{QPV}_{\mathrm{BB84}}^{f}$}

\definecolor{darkred}{RGB}{179, 16, 32}

\usepackage{slashed}
\usepackage{enumitem}
\usepackage{authblk}

\theoremstyle{plain}
\newtheorem{theorem}{Theorem}[section]

\newtheorem{definition}[theorem]{Definition}

\newtheorem{lemma}[theorem]{Lemma}

\usepackage[normalem]{ulem}

\usetikzlibrary{decorations.pathreplacing,calligraphy}

\usepackage[pdfencoding=auto, psdextra]{hyperref}
\hypersetup{
    bookmarksnumbered=true, 
    unicode=false, 
    pdfstartview={FitH}, 
    pdftitle={Security of a Continuous-Variable based Quantum Position
Verification Protocol}, 
    pdfauthor={Rene Allerstorfer, Llorenc Escola-Farras, Arpan Akash Ray, Boris Skoric, Florian Speelman and Philip Verduyn Lunel}, 
    colorlinks=true, 
    allcolors=blue,  
}

\title{Security of a Continuous\,-Variable based Quantum Position Verification Protocol
}
\author[1]{Rene Allerstorfer}
\author[1,2]{Lloren\c{c} Escol\`a-Farr\`as}
\author[3]{Arpan Akash Ray}
\author[3]{Boris \v{S}kori\'{c}}
\author[1,2]{Florian Speelman}
\author[1]{Philip {Verduyn Lunel}}
\date{\today}

\affil[1]{QuSoft, CWI Amsterdam, The Netherlands}
\affil[2]{Multiscale Networked Systems,  Informatics Institute, University of Amsterdam, The Netherlands }
\affil[3]{TU Eindhoven, The Netherlands }

\begin{document}

\maketitle

\begin{abstract}
\noindent In this work we study quantum position verification with continuous-variable quantum states. 
In contrast to existing discrete protocols, we present and analyze a protocol that utilizes coherent states and its properties. Compared to discrete-variable photonic states, coherent states offer practical advantages since they can be efficiently prepared and manipulated with current technology. 
We prove security of the protocol against any unentangled attackers via entropic uncertainty relations, showing that the adversary has more uncertainty than the honest prover about the correct response as long as the noise in the quantum channel is below a certain threshold. 
Additionally, we show that attackers who pre-share one continuous-variable EPR pair can break the protocol. 
\end{abstract}

\section{Introduction}\label{sec:Intro}

Position-based cryptography allows for protocols in which the geographical location of a party is used as a cryptographic credential. Consider, for example, the establishment of trust between you and someone at a claimed location. Or sending a confidential message that can only be decrypted at a specific location. Part of position-based cryptography is the task of position verification, where an untrusted prover aims to convince verifiers that he is present at a certain position $P$. 

This primitive was first introduced by Chandran, Goyal,
Moriarty, and Ostrovsky \cite{OriginalPositionBasedCryptChandran2009}, and it has been shown that no classical position verification protocol can exist, due to a universal attack based on cloning input information. This attack fails in the quantum setting because of the no-cloning theorem \cite{Wootters1982NoCloning}. Quantum position verification (QPV) has been studied\footnote{under the name of `quantum tagging'} since the early 2000s by several authors \cite{PatentKentANdOthers,Malaney_2010_b,Malaney_2010_a,Lau_2011}, but despite the failure of the classical universal attack, a universal quantum attack has since been found \cite{Buhrman_2014, Beigi_2011}. However, this attack consumes an amount of entanglement exponential in the input size and is therefore not practically feasible. Thus, we may still find secure QPV protocols in the bounded-entanglement model. 

The analysis of the entanglement resources needed turns out to be a deep question in its own right~\cite{Buhrman_2013,speelman2016instantaneous,dolev2022non,cree2022code,bluhm2022single,allerstorfer2023relating}. Many protocols have since been proposed~\cite{Chakraborty_2015,SWAP_protocol_Rene_et_all,gonzales2019bounds,liu_et_al:LIPIcs.ITCS.2022.100} and different security models have been studied~\cite{Unruh_2014_QPV_random_oracle,gao2016quantum,dolev2019constraining,allerstorfer2022role}. Recent work has focused on the practicality of implementing position-verification protocols. Aspects such as channel loss and error tolerance of certain QPV protocols must be taken into account~\cite{allerstorfer2022role,escolafarras2022singlequbit}.

Almost all previously studied QPV protocols have in common that they contain only finite-dimensional quantum systems. The study of QPV using continuous-variable (CV) quantum information, i.e., using infinite-dimensional quantum states, was first mentioned in \cite{qi2015loss}, in which a general attack was shown in the transmission regime $t \leq 1/2$, but the security of the protocol was not further analyzed.

The best known example of CV quantum information is the quantized harmonic oscillator \cite{cvqi_Braunstein_05, book_cerf_cvqi, cvqi_andersen_10}, which is usually described by continuous variables such as position and momentum. 
Continuous-variable quantum systems are particularly relevant for quantum communication and quantum-limited detection and imaging techniques because they provide a quantum description of the propagating electromagnetic field. Of particular relevance are the eigenstates of the annihilation operator, the so-called coherent states, and their quadrature squeezed counterparts known as squeezed coherent states.
The maiden appearance of CV quantum states in a quantum communication protocol was the CV variant of quantum key distribution (QKD).  Firstly proposed with discrete \cite{PhysRevA.61.010303, PhysRevA.61.022309, PhysRevA.62.062308} and Gaussian \cite{PhysRevA.63.052311} encoding of squeezed states, soon a variety of protocols were published on Gaussian-modulated CV-QKD with coherent states \cite{PhysRevLett.88.057902, article_Grossmman_03, 10.5555/2011564.2011570, PhysRevLett.93.170504}. In this paper, we employ many techniques borrowed from the wealth of research available on CV-QKD. Theoretical reviews with practical considerations of CV-QKD can be found in \cite{garcia2007quantum, leverrier:tel-00451021}.

We extend the ideas of finite-dimensional QPV protocols, and more formally analyze a QPV protocol very similar to the one mentioned in \cite{qi2015loss}. 
We provide a general proof of security against attackers who do not have access to entanglement, 
taking into account attenuation and excess noise in the quantum channel.
By way of illustration, we also analyze a number of specific attacks. We show that the attackers can break the scheme if they pre-share one pair of strongly entangled modes.

In the finite-dimensional case, usually the job of the prover is to complete a task \textit{correctly}, and attackers are detected by a suspiciously high error rate. 
This property of QPV protocols changes in the continuous setting, where even the honest prover's answers are drawn from a probability distribution. Therefore, the verifiers' job is to distinguish an honest sample from an adversarial one. 

Although the generalization of QPV to CV is interesting in itself, the motivation here is practical. 
CV systems are much simpler to handle in practice and leverage several decades of experience in
coherent optical communication technology. 
One particular advantage is that no true single-photon preparation or detection is necessary. 
Clean creation and detection of single photons is still expensive and technically challenging, especially if photon number resolution is desired. 
In contrast, homodyne and heterodyne measurements are easy to implement and a lot of existing infrastructure is geared towards handling light at low-loss telecom wavelengths (1310nm, 1550nm), whereas an ideal single photon source in these wavelength bands still has to be discovered and frequency up-conversion is challenging and introduces new losses and errors. 
Furthermore, loss causes a decrease in the signal-to-noise ratio
in homodyne measurements rather than giving a ``no detection'' event. 
This may open new avenues for protection against the usual lossy attack in discrete variable QPV protocols, in which attackers make use of the ``no detection'' rounds.

\section{Preliminaries}\label{sec:prelims}
In this section, we introduce the continuous-variable formalism that one encounters in CV-QKD, and some information-theoretic results. The goal of this section is threefold. First, we present the different types of CV states used in the paper. We then discuss displacement measurements that can be performed on these states and how a noisy channel is modeled. Finally, we close the section with some useful results from classical and quantum information theory.

\subsection{Gaussian states}
\label{sec:Gauss}
The Wigner function fully describes an $N$-mode bosonic quantum state $\rho$ and can be obtained from $\rho$ by the Wigner formula \cite{wigner1932quantum}
\begin{align}
    W(\mathbf{x}, \mathbf{p}) = \frac{1}{\pi^N} \int_{\mathbb{R}^N}
    e^{2i  \mathbf{p}\cdot \mathbf{y} } \langle \mathbf{x} - \mathbf{y} | \rho | \mathbf{x} + \mathbf{y} \rangle \, \mathrm{d}\mathbf{y}.
\end{align}
This is sometimes also called the Wigner transformation of the density matrix. The inverse transformation is archived via the Weyl transform. Gaussian states are defined by the property that their Wigner function is a Gaussian function in phase space. 
The Wigner function of Gaussian states reads
\begin{align}
    W_\mathrm{G}(\mathbf{r}) = \frac{1}{\pi^{N} \sqrt{\det \Gamma}} \exp \left\{ - (\mathbf{r} - \mathbf{d})^T \Gamma^{-1} (\mathbf{r} - \mathbf{d})\right\},
\end{align}
where $\mathbf{r} = (x_1, p_1, \dots, x_N, p_N)$ are the quadrature variables. The vector $\mathbf{d}$ is the displacement vector,
\begin{equation}
    d_i = \mathbb{E}\hat{r}_i = \Tr [\rho \hat{r}_i].
\end{equation}
And $\Gamma$ is the covariance matrix,
\begin{equation}
    \Gamma_{ij}=\Tr \left[ \rho \big( (\hat r_i-d_i)( \hat r_j-d_j) +(\hat r_j-d_j)( \hat r_i-d_i) \big) \right].
\end{equation}

\subsection{Displacement measurements of CV states}
\label{sec:CVmeas}
Here we describe homodyne and heterodyne measurements, the two types of possible displacement measurements.  For the physics of the measurement process, refer to Chapter 1 of \cite{garcia2007quantum}.

\subsubsection*{Homodyne}
Consider a Wigner function $ W(\mathbf{x}, \mathbf{p})$. A homodyne measurement of the quadrature $x_i$, yields the following marginal probability distribution
\begin{equation}
f_{X_i}(x_i)=\int_{\mathbb{R}^{2N-1}}W(\mathbf{x}, \mathbf{p}) \, \mathrm{d}\mathbf{p} \, \mathrm{d}x_1 \dots \mathrm{d}x_{i-1} \, \mathrm{d}x_{i+1} \dots \mathrm{d}x_{N}.
\end{equation}
One can choose any axis $x_\theta$ along which to perform a homodyne measurement, given a mode. In this case, we rotate our reference frame corresponding to the mode to be measured by an angle $\theta$. We can then perform an integral similar to the one above to obtain $f_{X_\theta}(x_\theta)$.

\subsubsection*{Heterodyne}
A heterodyne measurement is essentially a double homodyne measurement. The selected mode from $W(\mathbf{x}, \mathbf{p})$ is mixed with vacuum on a balanced beamsplitter. A homodyne measurement is then performed on the two output modes, each in conjugate directions. The result obtained is captured by the theorem which follows.
\begin{theorem}
    The heterodyne measurement of a one-mode Gaussian state with displacement $(x_0,p_0)$, produces two Gaussian distributions, centered around $x_0/\sqrt{2}$ and $-p_0/\sqrt{2}$ respectively.
\end{theorem}
\begin{proof}
    A balanced beamsplitter is represented by the following symplectic matrix \begin{eqnarray}
        S={\begin{pmatrix}
		\sqrt{\frac{1}{2}}\mathbbm{1}_2 & \sqrt{\frac{1}{2}}\mathbbm{1}_2 \\
	-\sqrt{\frac{1}{2}}\mathbbm{1}_2 & \sqrt{\frac{1}{2}}\mathbbm{1}_2
	\end{pmatrix}}.
    \end{eqnarray}
    As the input state is Gaussian, and mixing preserves Gaussian states, the output states are also Gaussian. The new displacements under this transformation are the given by 
    \begin{eqnarray}
        (x_0,p_0,0,0) S^T=(x_0/\sqrt{2},p_0/\sqrt{2},-x_0/\sqrt{2},-p_0/\sqrt{2}).
    \end{eqnarray}
\end{proof}

\subsection*{Noisy CV channel}

Whereas a discrete qubit state passing through a noisy channel
suffers from qubit loss, bit errors, and phase errors,
a continuous-variable state gets attenuated and acquires excess noise.
Consider a coherent state with displacement $(x_0,p_0)$.
Let $t\in[0,1]$ be the attenuation parameter,
and let $u\geq 0$ denote the excess noise power.\footnote{
In the CVQKD literature the excess noise power is often written as $\frac12 t\xi$,
where the proportionality with $t$ comes from the fact that the adversary mixes 
in his own quantum state using the same beamsplitter that also taps off part of the sender's state.
In our case we have no such adversarial action.
}
The effect of the channel is that the displacement becomes
$(x_0,p_0)\sqrt t$, and the covariance matrix goes from $\mathbbm{1}_2$ to $\mathbbm{1}_2(1+2u)$.
The outcome of a homodyne measurement now has the variance $\frac12+u$ instead of just the $\frac12$
from shot noise. 
In terms of signal and noise, the signal has changed by a factor $t$ and the noise has
increased by a factor $1+2u$. Overall,
the signal-to-noise ratio has changed by a factor $\frac t{1+2u}$.

\subsection{Continuous-variable EPR state and teleportation}
\label{sec:CVformalism}

Consider two modes labeled $A$ and $B$.
The Wigner function of the two-mode squeezed vacuum state (TMSV) with squeezing parameter $\zeta\geq 0$ is given by

\begin{equation}\label{eq Wigner fct EPR}
\begin{split}
    W_{\mathrm{TMSV}}(x_a,p_a,x_b,p_b)&=\frac{1}{\pi^2}\exp\{-e^{-2\zeta}[(x_a+x_b)^2+(p_a-p_b)^2]-e^{2\zeta}[(x_a-x_b)^2+(p_a+p_b)^2]\}\\&=\frac{1}{\pi^2}\exp\left\{-\begin{pmatrix}
        x_a&p_a&x_b&p_b
    \end{pmatrix}\Gamma(\zeta)^{-1}\begin{pmatrix}
 x_a\\
p_a\\
x_b\\
p_b
\end{pmatrix}\right\},
\end{split}
\end{equation}
with covariance matrix
\begin{equation}
    \Gamma(\zeta)=
\begin{pmatrix}
\cosh(2\zeta)\mathbbm{1}_2 & \sinh(2\zeta)Z \\
\sinh(2\zeta)Z & \cosh(2\zeta)\mathbbm{1}_2 \\
\end{pmatrix}, \qquad \mathrm{where} \qquad
Z=\begin{pmatrix}
1 & 0 \\
0 & -1 \\
\end{pmatrix}.
\end{equation}
Throughout this paper $\mathbbm{1}_n$ denotes the $n$-dimensional identity matrix. 
In the limit of the squeezing parameter $\zeta\rightarrow\infty$ we have $W_{\mathrm{TMSV}}(x_a,p_a,x_b,p_b)\rightarrow C\delta(x_a-x_b)\delta(p_a+p_b)$, for a constant $C$, which corresponds to the continuous-variable maximally entangled EPR state. 

Consider a heterodyne measurement performed on the $A$ mode.
The state of the $A$ mode, viewed in isolation, is a thermal state with covariance matrix $K_A=\mathbbm{1}_2 \cosh 2\zeta$.
Using a 50/50 beamsplitter this state gets mixed with the vacuum, resulting in a two-mode $A'A''$
state with covariance matrix
\begin{align}
    K_{A'A''}= \frac{1}{2}
    \begin{pmatrix}
        \mathbbm{1}_2 +K_A & \mathbbm{1}_2 -K_A \\
        \mathbbm{1}_2 -K_A & \mathbbm{1}_2 +K_A
    \end{pmatrix}
    =
    \begin{pmatrix}
        \mathbbm{1}_2\cosh^2 \zeta & -\mathbbm{1}_2\sinh^2 \zeta \\
        -\mathbbm{1}_2\sinh^2 \zeta & \mathbbm{1}_2\cosh^2 \zeta
    \end{pmatrix}.
\end{align}
In mode $A'$ the $x$-quadrature is measured, and in mode $A''$ the $p$-quadrature.
The Wigner function for $x_{a'}$ and $p_{a''}$ is obtained by integrating out $p_{a'}$ and $x_{a''}$
from the Wigner function $A'A''$, resulting in a product of two Gaussian distributions,
${\cal N}_{0,\frac{1}{2} \cosh^2\zeta}(x_{a'})  {\cal N}_{0,\frac{1}{2} \cosh^2\zeta}(p_{a''})$.

If the heterodyne measurement has resulted in $(x_{a'},p_{a''})$, then the post-measurement state of the $B$ subsystem
is a Gaussian state with displacement $(x_B,p_B)=(x_{a'},-p_{a''})\sqrt2 \tanh\zeta$ and covariance $\mathbbm{1}_2$, i.e.\ a coherent state
(see chapter 2 of \cite{leverrier:tel-00451021}).
Note that the components $x_B$ and $p_B$ are Gaussian-distributed with variance 
$\frac{1}{2}\cosh^2\zeta\cdot (\sqrt2 \tanh\zeta)^2=\sinh^2\zeta$.
In Section~\ref{sec:protocolEPR} we tune $\sinh\zeta=\sigma$ so that $x_B,p_B$ have Gaussian statistics
with variance $\sigma^2$. 

\subsubsection*{Teleportation}
The teleportation of an unknown continuous-variable quantum state using a CV EPR pair was proposed by Vaidman~\cite{vaidman1994teleportation} and is described as follows:
\begin{enumerate}
    \item Alice and Bob share a CV-EPR pair described by the Wigner function \eqref{eq Wigner fct EPR}. Alice possesses the single-mode quantum state $\ket\psi$ to be teleported.
    \item With a balanced beamsplitter, Alice mixes $\ket\psi$ with her mode of the CV-EPR pair and then does a measurement of the $x$-quadrature in one mode and the $p$-quadrature in the other mode (i.e.\ she performs a heterodyne measurement). We denote the outcome of the measurement as $(d_x,d_p)$. The result is that Bob's half of the EPR pair is transformed to a displaced version of~$\ket\psi$,
    with displacement $(\sqrt2 d_x,-\sqrt2 d_p)$. Alice sends the classical $(d_x,d_p)$ to Bob.
    \item Bob applies a displacement $(-\sqrt2 d_x,\sqrt2 d_p)$ to his state to obtain $\ket{\psi}$.
\end{enumerate}


\subsection{Information theory}
We now define some basic notions of information theory that will be used in the paper. First, we present some definitions and properties regarding CV entropies.
\begin{definition} Let $X$ be a continuous random variable with probability density function $f(x)$, and let $\mathcal{X}$ be its support set. The \emph{differential Shannon entropy} $h(X)$ is defined as
\begin{equation}
    h(X)=-\int_{\mathcal{X}}f(x)\log f(x) \, \mathrm{d}x,
\end{equation}
    where, if not otherwise mentioned, we use $log$ in base 2. 
\end{definition}

\begin{lemma}
\label{lemma:hscaling}
Let $\alpha >0$ and $X\in\mathbb R$. It holds that
$h(\alpha X) = h(X)+\log\alpha$.
\end{lemma}

\begin{definition}
The von Neumann entropy of a state $\rho$ is defined as 
$S(\rho)=-\operatorname{Tr} [\rho\log\rho]$.
\end{definition}
\noindent The von Neumann entropy of Gaussian states is provided by the following lemma, which will be needed to calculate entropies of the honest prover.
\begin{lemma}\label{lem: gfunc} \cite{PhysRevA.59.1820}
Let $\rho$ be an $N$-mode CV Gaussian state with $2N\times 2N$ covariance matrix $\Gamma$. 
Let $\nu_1,\ldots,\nu_N$ be the symplectic eigenvalues of~$\Gamma$.
Let the function $g$ be given by
\begin{equation}
    g(x)=(x+1)\log(x+1)-x \log x.
\label{defgfunction}
\end{equation}
The von Neumann entropy of $\rho$ is given by
\begin{equation}
    S(\rho) = \sum_{i=1}^N g\left(\frac{\nu_i-1}2\right).
\end{equation}
\end{lemma}

\begin{lemma} \cite{leverrier:tel-00451021}
The symplectic eigenvalue of a single-mode covariance matrix $\Gamma$ is given by
$\sqrt{\det \Gamma}$.
\end{lemma}
\noindent In Section~\ref{section:entropy_of_prover} we consider $\sigma \gg 1$ and are interested in the behavior of $g$ in that regime. The following lemma is not too hard to see from \eqref{defgfunction}.
\begin{lemma}
\label{lemma:glarge}
The large-argument behavior of the function $g$, defined in (\ref{defgfunction}), 
is given by
$g(x)\sim \log (ex) +{\cal O}(1/x)$.
\end{lemma}

\noindent Another useful quantity to compare two quantum states is the relative entropy.
\begin{definition}
    Let $\rho$ and $\sigma$ be two density matrices. 
    Their \emph{Umegaki's quantum relative entropy} $D(\cdot||\cdot)$ is defined as
    \begin{equation}
        D(\rho||\sigma)=\operatorname{Tr} \left[ {\rho\log\rho}-{\rho\log\sigma} \right].
    \end{equation}
\end{definition}

\noindent As introduced in \cite{furrer2014position}, let $\rho_{AB}$ be a bipartite state on systems $A$ and $B$, which correspond to a system to be measured and a system held by an observer. Let $X$ be a continuous random variable, $\alpha=2^{-n}$ for some $n\in\mathbb N$, and consider the intervals $\mathcal{I}_{k;\alpha}:=(k\alpha,(k+1)\alpha]$ for $k\in\mathbb Z$.  Here $\rho_B^{k;\alpha}$ denotes the sub-normalized density matrix in $B$ when $x$ is measured in $\mathcal{I}_{k;\alpha}$, $\rho_B^x$ denotes the conditional reduced density matrix in $B$ so that $\int_{\mathcal{I}_{k;\alpha}}\rho_B^xdx=\rho_B^{k;\alpha}$, and $Q_{\alpha}$ denotes the random variable that indicates which interval $x$ belongs to. These notions are used in the continuous version of the conditional entropy.

\begin{definition}
    The \emph{quantum conditional von Neumann entropy} is defined as 
    \begin{equation}
        H(Q_{\alpha}|B)_\rho:=-\sum_{k\in\mathbb Z}D(\rho^{k;\alpha}_{B}||\rho_B).
    \end{equation}
\end{definition}

\begin{definition} We define the  \emph{differential quantum conditional von Neumann entropy} is defined as
\begin{equation}
    h(X|B)_{\rho} := -\int_{\mathbb R} D(\rho_B^x || \rho_B) \, \mathrm{d}x. 
\end{equation}
\end{definition}
    
\noindent The basis of our security proofs is the quantum-mechanical uncertainty principle. We use the following form for the differential entropy in a tripartite setting of a guessing game, as is often useful in the context of quantum cryptography. 

\begin{lemma} 
\label{lemma:uncertainty}
\cite{furrer2014position} 
Let $\rho_{ABC}$ be a tripartite density matrix on systems $A$, $B$ and $C$. Let $Q$ and $P$ denote the random variables of position and momentum respectively, resulting from a homodyne measurement on the $A$ system and let the following hold: $h(Q|B)_{\rho}, h(P|C)_{\rho}>-\infty$ and $H(Q_{\alpha}|B)_{\rho}, H(P_{\alpha}|C)_{\rho}<\infty$ for any $\alpha>0$. Then
\begin{equation}
    h(Q|B)_{\rho}+h(P|C)_{\rho}\geq\log(2\pi).
\end{equation}
\end{lemma}

\noindent Furthermore, we will make use of the following estimation inequality.

\begin{theorem} \cite{cover1999elements}\label{theorem:fano}
    Let $X$ be a random variable and $\hat{X}(Y)$ an estimator of $X$ given side information $Y$, then
    \begin{equation}
        \mathbb{E} \left[ \left(X-\hat{X}(Y)\right)^2 \right] \geq \frac{1}{2\pi e}e^{2h_{\mathrm{nats}}(X|Y)},
    \end{equation}
    where $h_{\mathrm{nats}}(X|Y)$ is the conditional entropy in natural units. Moreover, if $X$ is Gaussian and $\hat{X}(Y)$ is its mean, then the equality holds. 
\end{theorem}

\section{The Protocol}

\subsection{Prepare-and-measure}\label{sec:protocol_pm}

Consider two spatially separated verifiers $V_1$ and $V_2$, and a prover P somewhere in between them. 
Let $\cal A$ be a publicly known set of angles in $[0,2\pi)$ such that
$\alpha \in {\cal A} \implies \alpha+\pi/2\in{\cal A}$.
Let $\sigma$ be a publicly known parameter, $\sigma\gg 1$. 
A single round of the protocol consists of the following steps (for a diagrammatic picture see Fig.\,\ref{fig:protocol}):
\begin{enumerate}
\item The verifiers draw random $\theta\in\cal A$ and two random variables $(r,r^\perp)$ from the Gaussian distribution ${\cal N}_{0,\sigma^2}$.
Verifier $V_1$ prepares a coherent state $\ket\psi$ 
with quadratures $(x_0,p_0) = (r\cos\theta + r^\perp\sin\theta, \; r\sin\theta-r^\perp\cos\theta)$.
Then $V_1$ sends $\ket\psi$ to the prover, and $V_2$ sends $\theta$ to the prover.

\item The prover receives $\theta$ and $\ket\psi$ and
performs a homodyne measurement on $\ket\psi$ in the $\theta$ direction,
resulting in a value $r'\in{\mathbb R}$.
The prover sends $r'$ to both verifiers.
\end{enumerate}

\noindent After $n$ rounds, the verifiers have received a sample of responses, which we denote as $(r_i')_{i=1}^n$. The verifiers check whether all prover responses arrived at the correct time,
and whether the reported values $(r_i')_{i=1}^n$  satisfy
\begin{equation}
    \frac1n\sum_{i=1}^n \frac{\left(r_i'-r_i\sqrt t\right)^2}{\frac12+u} <  \gamma \quad\quad \mbox{with }  
    \gamma\stackrel{\rm def}{=} 1 +\frac2{\sqrt n}\sqrt{\ln\frac1{\eps_{\rm hon}}} + \frac2n\ln\frac1{\eps_{\rm hon}}.
\label{score}
\end{equation}
Here $\eps_{\rm hon}$ is an upper bound on the honest prover's failure probability, see Section~\ref{sec:honestprover}. $\eps_{\rm hon}$ is a protocol parameter and can be set to a desired value. The verifiers {\it reject} if not all these checks are satisfied. We refer to the sum in (\ref{score}) as the {\em score}.

\begin{figure}[ht]
    \centering
    \begin{tikzpicture}[node distance=3cm, auto]
    \node (A) {$V_1$};
    \node [left=1cm of A] {};
    \node [right=of A] (B) {$P$};
    \node [right=of B] (C) {$V_2$};
    \node [right=1cm of C] {};
    \node [below=of A] (D) {};
    \node [below=of B] (E) {}; 
    
    \node [right=1.5cm of A] (M) {};
    \node [left=1.5cm of C] (M2) {};
    
    \node [below=of C] (F) {};
    \node [below=of D] (G) {$V_1$};
    \node [below=of E] (H) {};
    \node [below=of F] (I) {$V_2$};
    \node [left= 6cm of E] (J) {};
    \node [below= 3cm of J] (K) {};
    \node [above= 3cm of J] (L) {};

    \draw [->, transform canvas={xshift=0pt, yshift = 0 pt}, quantum] (A) -- (E) node[midway] (x) {} ;
  
    \draw [->] (C) -- (E);
    \draw [][->] (E) -- (I) node[midway] (q) {$r'$}; 
    \draw [][->] (E) -- (G); 

    \draw [->] (L) -- (K) node[midway] {time};

    \node[left=1cm of x, transform canvas={xshift=+ 2pt, yshift = +2 pt}] {coherent state $\ket{\psi}$};
    \node[right = 3.5cm of x, transform canvas={xshift=+ 2pt, yshift = +2 pt}] {$\theta \in {\cal A}$};
    \node[left = 3.3cm of q] {$r'$};

    \node [above=0.5cm of A] (posV00) {};
    \node [left=1cm of posV00] (posV0) {};
    \node [above=0.5cm of C] (posV11) {};
    \node [right=1cm of posV11] (posV1) {};
    \draw [->] (posV0) -- (posV1) node[midway] {position};

    \node [right=0.8cm of C] (VP0) {};
    \node [right=4.55cm of E] (VPP) {};
    \node [right=0.8cm of I] (PV) {};
\end{tikzpicture}
\caption{Schematic representation of the protocol described in Section~\ref{sec:protocol_pm}. Undulated lines represent quantum information, whereas straight lines represent classical information.}
\label{fig:protocol}
\end{figure}
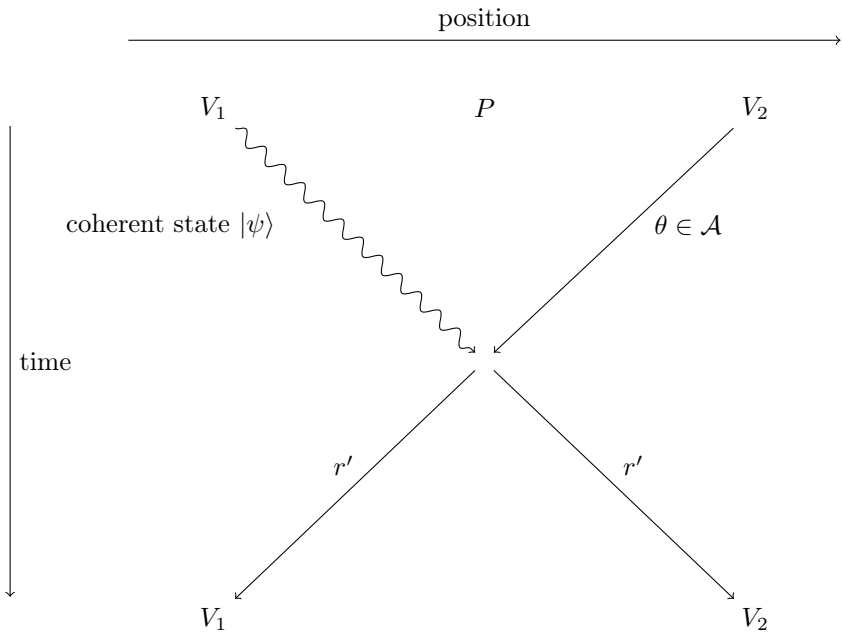

\subsection{Entanglement based version of the protocol}
\label{sec:protocolEPR}

In security proofs for qubit-based schemes, it is customary to re-formulate a protocol into an EPR based form. The act of one party ($V$) preparing and sending a qubit state in a particular basis $\cal B$ is equivalent to $V$ preparing a maximally entangled two-qubit state (EPR pair) and then measuring one of the qubits in the $\cal B$ basis while the other qubit is sent. The act of measuring can be postponed. This has the advantage that in the security analysis the basis choice can be delayed, and it is then possible to base security on the properties of entangled states.

We will do an analogous reformulation for CV states.
In fact, we work with exactly the same states as Gaussian-modulated CV-QKD \cite{laudenbach2018continuous}.
We tune the squeezing parameter $\zeta$ such that $\sinh\zeta=\sigma$, as explained in Section~\ref{sec:CVformalism}.
Preparing a coherent state with Gaussian distributed displacements $x_0,p_0\sim {\cal N}_{0,\sigma^2}$
is equivalent to preparing a two-mode squeezed state with squeezing parameter $\zeta$ and then performing
a heterodyne $(\hat x,\hat p)$ measurement on one mode, with measurement outcome 
$\frac{(x_0, -p_0)}{\sqrt2 \tanh\zeta}$.

In our particular case, the verifier $V_1$ prepares the two-mode squeezed state $\rho_{VP}$ and performs the heterodyne measurement with quadratures that are rotated by the angle~$\theta$ on the $V$ subsystem.
The measurement outcomes are $r/(\sqrt2 \tanh\zeta )$ and $-r^\perp/(\sqrt2 \tanh\zeta)$, resulting in displacement
$(r,r^\perp)$ in the state sent to the prover (i.e.\ subsystem $P$). The prover then performs a homodyne measurement
under angle $\theta$ to recover $r$, similar to the prepare and measure scheme.

In the security analysis in Section~\ref{sec:generalattack} we will explicitly write $V_1$'s heterodyne measurement as 
a double-homodyne measurement. 
First $V_1$ mixes its own mode with the vacuum using a beamsplitter, resulting in a two-mode state. 
On one of these modes $V_1$ then performs a homodyne measurement in the $\theta$-direction, on the other mode in the $\theta+\frac\pi2$ direction.

\subsection{Honest prover} \label{sec:honestprover}

\subsubsection{Success probability}

We show that the honest prover has a failure probability smaller than~$\eps_{\rm hon}$.

\begin{lemma}
\label{lemma:chisquare}
(Eq.(4.3) in \cite{laurent2000adaptive})
Let $X$ be a $\chi^2_n$ distributed random variable. 
It holds that
\begin{equation}
    \mathbb{P}[X-n \geq 2\sqrt{n a}+2a] \leq e^{-a}. 
\end{equation}
\end{lemma}

\noindent In round $i$, the honest prover performs a homodyne measurement under an angle $\theta_i$, 
on a coherent state that has displacement $r_i$ in the $\theta_i$ direction (and displacement $r_i^\perp$ in the $\theta_i+\frac\pi2$ direction).
The measurement outcome $R_i'$ is Gaussian-distributed with mean $r_i$ and variance $\frac{1}{2}$ (shot noise).
The random variable $Z=\sum_{i=1}^n (R_i'-r_i\sqrt t)^2/(\frac12+u)$ is chi-square distributed with parameter~$n$, i.e. $Z\sim\chi^2_n$.
The probability that the honest prover fails to pass verification is given by
\begin{equation}
\mathbb{P}[Z\geq  n\gamma] =\mathbb{P} \left[Z\geq n+2\sqrt{n\ln\frac1{\eps_{\rm hon}}} +2\ln\frac1{\eps_{\rm hon}} \right].
\end{equation}
By Lemma~\ref{lemma:chisquare} this is upper bounded by~$\eps_{\rm hon}$.

\subsubsection{A posteriori distribution and entropy of \texorpdfstring{$R$}{R} conditioned on measurement}

We determine how much uncertainty the honest prover has about the displacements $r_i$,
given the measurement outcomes~$r_i'$.
For notational brevity we omit the round index~$i$.
We write the probability density for $R$ as $f_R$. Since $r'$ is the result of a measurement under angle $\theta$, 
conditioning on $\theta$ is implicit and will be omitted from the notation.

The prover's posterior distribution of $R$, given $r'$, is 
\begin{equation}
    f_{R|R'}(r|r') = \frac{f_{RR'}(r,r')}{f_{R'}(r')} = \frac{f_R(r)f_{R'|R}(r'|r)}  {f_{R'}(r')}.
\end{equation}
Using $f_R={\cal N}_{0,\sigma^2}$,
$f_{R'|R}(r'|r)={\cal N}_{r\sqrt t,\frac12+u}(r')$
and 
$f_{R'}={\cal N}_{0,t\sigma^2+\frac{1}{2} +u }$
we get, after some algebra,
\begin{equation}
    f_{R|R'}(r|r') = {\cal N}_{M,\Sigma^2}(r)
    \qquad
    \mbox{with } \Sigma^2 \stackrel{\rm def}{=} \left(\frac{1}{\sigma^2}+\frac{t}{1/2+u}\right)^{-1},
    \quad
    M \stackrel{\rm def}{=} \frac{r'}{\sqrt t}\cdot \frac1{1+\frac{1/2+u}{t\sigma^2}}.
\label{proverposterior}
\end{equation}
For $t\sigma^2\gg 1$ this tends to a normal distribution centered on $r'/\sqrt t$, with 
variance~$(\frac{1}{2}+u)/t$. From the Gaussian probability density function (\ref{proverposterior})
we directly obtain the differential entropy of $R$ given~$R'$,
\begin{equation}
\label{entropy prover}
    h(R|R') = \frac{1}{2} \log 2\pi e\Sigma^2.
\end{equation}

\subsubsection{Entropy of \texorpdfstring{$R$}{R} conditioned on the prover's quantum state}\label{section:entropy_of_prover}

Let $\rho_{VP}$ be the entangled state that is prepared by $V_1$ as described in Section~\ref{sec:protocolEPR}.
Here $P$ denotes the prover's quantum system.
The heterodyne measurement on the $V$ system yields $(r,r^\perp)$.
The measurement maps $\rho_{VP}$ to $\rho_{RR^\perp P}$.
We write the post-measurement state as
\begin{equation}
    \rho_{RR^\perp P}=\int_{\mathbb{R}^2}  f_{RR^\perp}(r,r^\perp)\ketbra{r}{r}_R\otimes\ketbra{r^\perp}{r^\perp}_{R^\perp}\otimes\rho^{rr^\perp}_P \, \mathrm{d}r \mathrm{d}r^\perp.
\end{equation}
The (differential) entropy of $R$, conditioned on the prover's quantum state, can be expanded as
\begin{equation}
    h(R|P) = h(R)+S(P|R) - S(P).
\end{equation}
From the definition of conditional entropy, $S(P|R)=\mathbb{E}_rS(\mathbb{E}_{r^\perp}\rho^{rr^\perp}_P)$ and $S(P)=S(\mathbb{E}_{r,r^\perp}\rho^{rr^\perp}_P)$. As discussed in Section~\ref{sec:protocolEPR},  $\rho^{rr^\perp }_{P}$ would be a coherent state in an ideal case. However, in a noisy channel, the state becomes a Gaussian with covariance matrix $(1+2u)\mathbbm{1}_2$ and displacement $(x_0\sqrt{t},p_0\sqrt{t})$. The expectations are Gaussian integrals and hence are exactly solvable. Solving these integrals, we end up with the corresponding Gaussian Wigner functions and symplectic eigenvalues

\begin{align}
    \text{for } \mathbb{E}_{r^\perp}\rho^{rr^\perp}_P: & \quad W_r(x,p) \sim \exp\left(-\frac{p^2}{2 \sigma ^2 t+2 u+1}-\frac{\left(x-r \sqrt{t}\right)^2}{2 u+1}\right), \quad \nu=\sqrt{2t\sigma^2+2u+1},\\
    \text{for } \mathbb{E}_{r,r^\perp}\rho^{rr^\perp}_P: & \quad W(x,p) \sim \exp -\frac{p^2+x^2}{2 \sigma ^2 t+2 u+1}, \quad \nu=2t\sigma^2+2u+1.
\end{align}
We use the $g$ function (cf. lemma~\ref{lem: gfunc}) to calculate the corresponding entropy
\begin{align}
    S(P|R)&=g\left(\frac{1}{2}\sqrt{2t\sigma^2+2u+1}-\frac{1}{2}\right),\\
    S(P)&= g(t\sigma^2+u).
\end{align}
Finally, by definition $R$ is Gaussian and $h(R)=\frac{1}{2}\log 2\pi e\sigma^2$. All together this yields
\begin{equation}
    h(R|P) =\frac{1}{2} \log 2\pi e \sigma^2 + g\left(\frac{1}{2}\sqrt{2t\sigma^2+2u+1}-\frac{1}{2}\right) -g(t\sigma^2+u)
    \stackrel{ {\rm Lemma}~\ref{lemma:glarge} }{=}
    \frac{1}{2}\log \frac{\pi e t}{1+2u} + O\left(\frac1\sigma\right).
\end{equation}
For large $\sigma$ this is essentially the same as $h(R|R')$ in (\ref{entropy prover}). 


\section{Security against specific attacks}

Before showing security against a general attack, we highlight security against some specific attacks that one might naturally think of. We look into three specific attacks where the adversaries do not have access to entanglement: performing a heterodyne measurement, state splitting and performing a homodyne measurement under a guessed angle.
These examples provide some insight into the security but do not constitute a general security proof.
A rigorous security proof for the case of adversaries who do not pre-share entanglement is given in Section~\ref{sec:generalattack}.

The most general attack of a 1-dimensional QPV protocol consists of placing two attackers Alice~$A$ and Bob~$B$ between $V_1$ and $P$, and $V_2$ and $P$, respectively. For attackers that do not pre-share entanglement,\footnote{We restrict our analysis to the case where the attackers do not pre-share entanglement, since we show in Section~\ref{sec:EPRattack} that there exists a perfect attack if they pre-share an EPR pair.} an attack proceeds as follows. Alice intercepts the quantum state sent to the prover $P$. Alice applies a local operation to her quantum system and sends some classical and/or quantum information to the second attacker Bob. The most general action Bob can take is to intercept the message $\theta$ and broadcast it, since any quantum operation can be embedded in Alice's actions.
After one round of simultaneous communication, Alice and Bob use their respective quantum and classical information to produce a classical output and respond to their closest verifier such that the answer arrives on time. For the following analysis, we describe the attacks per round of the protocol.

\subsection{Heterodyne attack}

In a \emph{heterodyne attack},
Alice performs a heterodyne measurement on the coherent state she intercepts and sends the result $(x',p')$ to Bob. 
At the end, $A$ and $B$ report the best guess for $r$ that they can produce based on $x',p',\theta$.
Let us denote this as the estimator $\tilde r$.
It holds that $\tilde r=\tilde x \cos\theta+\tilde p \sin\theta$,
where $\tilde x$ is an estimator for $x_0$, and similarly $\tilde p$.
The posterior distribution of $x_0$ given $x'$ is
\begin{equation}
    f_{X_0|X'}(x_0|x') = \frac{f_{X_0}(x_0) f_{X'|X_0}(x'|x_0) }  {f_{X'}(x') }
    = \frac{{\cal N}_{0,\sigma^2}(x_0) {\cal N}_{\frac{x_0}{\sqrt2},\frac{1}{2}}(x') }  {{\cal N}_{0,\frac{\sigma^2}2+\frac{1}{2}}(x')}
    = {\cal N}_{x' \sqrt2 \frac{\sigma^2}{1+\sigma^2}, \frac{\sigma^2}{1+\sigma^2}}(x_0).
\label{posteriorhetero}
\end{equation}
Hence $\tilde x=x' \sqrt2 \frac{\sigma^2}{1+\sigma^2}$ and $\tilde p=-p' \sqrt2 \frac{\sigma^2}{1+\sigma^2}$.
Given $x_0,y_0,\theta$, the random variable $\tilde R$ is Gaussian with mean $\frac{\sigma^2}{1+\sigma^2}r$
and variance
$\frac{1}{2} (\sqrt2 \frac{\sigma^2}{1+\sigma^2} \cos\theta)^2 + \frac{1}{2} (\sqrt2 \frac{\sigma^2}{1+\sigma^2} \sin\theta)^2$
$=(\frac{\sigma^2}{1+\sigma^2})^2$.
This gives
\begin{equation}
    \mathbb E(\tilde R-r)^2=\left(\frac{\sigma^2}{1+\sigma^2}\right)^2 + r^2 \left(\frac1{1+\sigma^2}\right)^2
    \approx 1 \qquad \text{for } \sigma \gg 1, 
\label{EdistHetero}
\end{equation}
which is easily distinguishable from the honest prover's value $\frac{1}{2}$.\footnote{
Note that the unbiased estimator $x'\sqrt2 \cos\theta-p'\sqrt2\sin\theta$
would yield $\mathbb E(\tilde R-r)^2=1$, which is larger than (\ref{EdistHetero}).
}
From (\ref{posteriorhetero}) we obtain the variance of $R$ from the attackers' point of view as
$\frac{\sigma^2}{1+\sigma^2}(\cos\theta)^2 +\frac{\sigma^2}{1+\sigma^2}(\sin\theta)^2$
$=\frac{\sigma^2}{1+\sigma^2}$.
The attackers' ignorance about $R$ is thus quantified as
\begin{equation}
    h(R|X' P' \Theta) = \frac{1}{2} \log \left(2\pi e \frac{\sigma^2}{1+\sigma^2}\right),
\end{equation}
with conditioning on $\Theta$ being made explicit.

\subsection{Splitting attack}
In a \emph{splitting attack}, Alice intercepts the coherent quantum state sent by $V_1$, and as in the case of the previous attack, she used a beamsplitter to mix it with a state of her own. She now sends one of the outputs from the beamsplitter to Bob. This allows both attackers to perform a homodyne measurement under the correct angle $\theta$. Unlike the heterodyne attack, this also allows the attackers the freedom to choose the transmittance parameter $T$ and the quantum state that Alice uses. However, the attackers must be cautious to report a set of numbers that have identical means and variances. To see why, let us assume that Alice reports numbers with mean $m_a$ and Bob's results have the mean $m_b$. Let the respective variances be $v_a=v_b$. The verifiers can immediately identify an attack if the results have a dissimilar average. To avoid this, Alice (or Bob) must multiply their results with a finite number $ c$ such that $m_b=cm_a$ (or $m_a=cm_b$). However, it would lead to the verifiers possessing a final distribution with indeed the same mean, but different variances.  The precision of the protocol can be altered to detect said variance. A similar argument can be constructed when the variances are unequal. Thus, a successful attack must follow $m_a=m_b$ and $v_a=v_b$. Now, we propose the following theorem.
 
\begin{theorem}

Consider a 2-mode Gaussian Wigner function $W_{\mathbf{d},\gamma}(x_1,p_1,x_2,p_2)$ which under a beamsplitter transformation of transmittance $T$ transforms into $W'_{\mathbf{d}',\gamma'}(x'_1,p'_1,x'_2,p'_2)$. 
If $\abs{\mathbb{E}[r'_1]}=\abs{\mathbb{E}[r'_2]}$ and $\mathrm{var}(r'_1)=\mathrm{var}(r'_2)$, then $\mathbf{d}_2=0$ and $T=1/2$, for $r \in \{x,p\}$. Here, $\mathbf{d}=(\mathbf{d}_1,\mathbf{d}_2)$ and $\mathbf{d}'=(\mathbf{d}'_1,\mathbf{d}'_2)$.

\end{theorem}

\begin{proof}
We have the following relationship between the covariance matrices of the input and output states
\begin{equation}
\gamma'=S\gamma S^T.
\end{equation}
Where $S$ is the symplectic matrix corresponding to a beamsplitter with transmittance T given by
\begin{equation}
S={\begin{pmatrix}
		\sqrt{T}\mathbbm{1}_2 & \sqrt{1-T}\mathbbm{1}_2 \\
	-\sqrt{1-T}\mathbbm{1}_2 & \sqrt{T}\mathbbm{1}_2
	\end{pmatrix}}.
\end{equation}
The input matrix $\gamma$ is the direct sum of the constituent matrices,
\begin{equation}
\gamma=\gamma_1 \oplus \gamma_2.
\end{equation}
Assuming some displacements $\mathbf{d}'$, we calculate the exponent in $W_{\mathbf{d}',\gamma'}$,
\begin{equation}
	(\mathbf{r}_1-\mathbf{d}'_1,\mathbf{r}_2-\mathbf{d}'_2) \gamma'^{-1} (\mathbf{r}_1-\mathbf{d}'_1,\mathbf{r}_2-\mathbf{d}'_2)^T.
\end{equation}
Substituting, after some matrix multiplications
\begin{align}
	&(\mathbf{r}_1-\mathbf{d}'_1,\mathbf{r}_2-\mathbf{d}'_2) \gamma'^{-1} (\mathbf{r}_1-\mathbf{d}'_1,\mathbf{r}_2-\mathbf{d}'_2)^T\\&=(\mathbf{r}_1-\mathbf{d}'_1,\mathbf{r}_2-\mathbf{d}'_2) \frac{1}{D}\begin{pmatrix}
		(T\gamma_2+(1-T)\gamma_1)\mathbbm{1}_2 & \sqrt{T(1-T)}(\gamma_2-\gamma_1)\mathbbm{1}_2 \\
		\sqrt{T(1-T)}(\gamma_2-\gamma_1)\mathbbm{1}_2 & (T\gamma_1+(1-T)\gamma_2)\mathbbm{1}_2
	\end{pmatrix} \begin{pmatrix}
	\mathbf{r}_1-\mathbf{d}'_1 \\
	\mathbf{r}_2-\mathbf{d}'_2
	\end{pmatrix}\\
	&=\frac{1}{D}\left( (T\gamma_2+(1-T)\gamma_1)(\mathbf{r}_1-\mathbf{d}'_1)^2+(T\gamma_1+(1-T)\gamma_2)(\mathbf{r}_2-\mathbf{d}'_2)^2\right) ,\label{post_bs_wigner_splitatk}
\end{align}
where $D$ is the determinant of $\gamma'$. We are given that $\mathrm{var}(r'_1)=\mathrm{var}(r'_2)$. From the construction of the Wigner function, it is clear that the coefficients in (\ref{post_bs_wigner_splitatk}) must be identical for this to be true, so
\begin{equation}
	T\gamma_2+(1-T)\gamma_1=T\gamma_1 + (1-T)\gamma_2 \Rightarrow T=1/2.
\end{equation}
The displacement transforms as 
\begin{equation}
	(\mathbf{d}'_1,\mathbf{d}'_2)=(\mathbf{d}_1,\mathbf{d}_2)S^T=(\sqrt{T}\mathbf{d}_1+\sqrt{1-T}\mathbf{d}_2, -\sqrt{1-T}\mathbf{d}_1+\sqrt{T}\mathbf{d}_2).
\end{equation}
As $\abs{\mathbb{E}[r'_1]}=\abs{\mathbb{E}[r'_2]}$ (or $\abs{\mathbf{d}'_1}=\abs{\mathbf{d}'_2}$), 
\begin{equation}
	\abs{\sqrt{T}\mathbf{d}_1+\sqrt{1-T}\mathbf{d}_2}=\abs{-\sqrt{1-T}\mathbf{d}_1+\sqrt{T}\mathbf{d}_2}.
\end{equation}
The only meaningful case from this equation yields
\begin{equation}
	\mathbf{d}_2=\frac{\sqrt{T}-\sqrt{1-T}}{\sqrt{T}+\sqrt{1-T}}\mathbf{d}_1.
\end{equation}
When $T=1/2$, this leads to $\mathbf{d}_2= \mathbf{0}$.
\end{proof}

The above theorem fixes the displacement and the transmittance parameter.  However, as we see, there is no restriction on the attackers for choosing the covariance matrix for their quantum state. Since the strongest attack must have the smallest spread, the natural choice is indeed the minimum uncertainty state, that is, a state with unit covariance. 

Hence, the strongest attack is carried out by mixing a vacuum state with the target using a balanced beamsplitter. 

After mixing, $A$ and $B$ have a coherent state with displacement $\frac{(x_0,p_0)}{\sqrt2}$ and $-\frac{(x_0,p_0)}{\sqrt2}$ respectively. 
Taking into account that $B$ compensates for the minus sign,
a homodyne measurement under the correct angle $\theta$ yields an outcome $u$ with distribution
$f_{U|R}(u|r)={\cal N}_{\frac r{\sqrt2},\frac{1}{2}}(u)$
for both attackers.
Their a posteriori distribution for $R$ is
\begin{equation}
    f_{R|U}(r|u)=\frac{f_R(r) f_{U|R}(u|r) }{f_U(u)}
    =\frac{{\cal N}_{0,\sigma^2}(r)  {\cal N}_{\frac r{\sqrt2},\frac{1}{2}}(u) }   {{\cal N}_{0,\frac{\sigma^2}2+\frac{1}{2}}(u)},
\end{equation}
which is the same as for the heterodyne attack.
The rest of the analysis is identical to that case.

\subsection{Attackers perform a homodyne measurement under a guessed angle}

In this attack, Alice picks a random angle $\varphi$ and does a homodyne measurement under this angle.
She forwards the result $m$ to Bob.
The distribution of $m$ is given by
$f_{M|X_0 P_0 \Phi}(m|x_0 p_0 \varphi) = {\cal N}_{x_0\cos\varphi+p_0\sin\varphi,\frac{1}{2}}(m)$
=${\cal N}_{r\cos(\varphi-\theta)+r^\perp\sin(\varphi-\theta),\frac{1}{2}}(m)$.
The attackers' posterior distribution for $R$ is
\begin{eqnarray}
    f_{R|M\Phi\Theta}(r|m\varphi\theta) &=&
    \frac{f_\Theta(\theta)f_\Phi(\varphi) f_R(r) f_{M|R\Theta\Phi}(m|r\theta\varphi)}
    {f_\Theta(\theta)f_\Phi(\varphi) f_{M|\Theta\Phi}(m|\theta\varphi) }
    \\ & \propto &
    f_R(r) f_{M|R\Theta\Phi}(m|r\theta\varphi)
    \\ &=& 
    f_R(r) {\mathbb E}_{r^\perp} f_{M|RR^\perp\Theta\Phi}(m|rr^\perp\theta\varphi) 
    \\ &=&
    f_R(r) {\mathbb E}_{r^\perp} {\cal N}_{r\cos(\varphi-\theta)+r^\perp\sin(\varphi-\theta),\frac{1}{2}}(m)
    \\ &=& 
    {\cal N}_{0,\sigma^2}(r) {\cal N}_{r\cos(\varphi-\theta),\frac{1}{2}+\sigma^2 \sin^2(\varphi-\theta)}(m).
\end{eqnarray}
After some algebra this can be rewritten as
\begin{equation}
    f_{R|M\Phi\Theta}(r|m\varphi\theta) = 
    {\cal N}_{\mu,S^2}(r)
    \quad \mbox{ with }
    \mu=m\cos(\varphi-\theta)\frac{\sigma^2}{\frac{1}{2}+\sigma^2}, \quad
    S^2=\sigma^2 \frac{\frac{1}{2}+\sigma^2\sin^2(\varphi-\theta)}{\frac{1}{2}+\sigma^2}.
\end{equation}

\noindent The attackers send $\mu$ to the verifiers. For the expected score we get 
\begin{eqnarray}
    {\mathbb E} (R-\mu)^2 &= & 
    {\mathbb E} R^2 + {\mathbb E} \mu^2 -2 {\mathbb E} \mu R
    \\ &=&
    \sigma^2 + \left(\frac{\sigma^2}{\frac{1}{2}+\sigma^2}\right)^2 {\mathbb E} m^2 \cos^2(\varphi-\theta) 
    -2\frac{\sigma^2}{\frac{1}{2}+\sigma^2}{\mathbb E} mr\cos(\varphi-\theta).
\label{interimRandomPhi}
\end{eqnarray}
We introduce the notation $\delta=\varphi-\theta$.
We use the distribution of $m$ conditioned on $rr^\perp \varphi\theta$ to write
\begin{eqnarray}
    {\mathbb E} \cos^2 \delta\; m^2  & =&
    {\mathbb E} \cos^2\delta \left[\frac{1}{2} + (r\cos\delta+r^\perp\sin\delta)^2\right]
    \\ &=&
    \frac{1}{2} {\mathbb E} \cos^2\delta + ({\mathbb E} r^2){\mathbb E} \cos^4 \delta+({\mathbb E} [r^\perp]^2) {\mathbb E}\cos^2\delta \sin^2\delta
    \\ &=&
    \frac{1}{2} {\mathbb E} \cos^2\delta + \sigma^2 {\mathbb E} \cos^2\delta
    \\ &=&
    \frac{1}{2}\left(\frac{1}{2}+ \sigma^2\right).
\label{Ec2m2}
\end{eqnarray}
Here we have used that ${\mathbb E} \cos^2\delta=\frac{1}{2}$ because of the uniform~$\varphi$.
Furthermore we have
\begin{equation}
    {\mathbb E}  m r \cos\delta = {\mathbb E} (r^2 \cos^2 \delta+rr^\perp \sin\delta\cos\delta) = \sigma^2/2 +0.
\label{Emrc}
\end{equation}
Substitution of (\ref{Ec2m2},\ref{Emrc}) into (\ref{interimRandomPhi}) yields
\begin{equation}
    {\mathbb E} (R-\mu)^2 = \frac{\sigma^2}2\cdot\frac{\sigma^2+1}{\sigma^2+\frac{1}{2}}.
\end{equation}
This is much larger than the honest prover's value~$1/2+u$ for sufficiently large $\sigma$.

\section{Security against general attacks by unentangled adversaries}
\label{sec:generalattack}

In this section, we show that we not only have security against the above described attacks, but that the result generalizes to all attackers that do not pre-share entanglement by lower bounding their uncertainty higher than the prover's. This is captured by the following theorem.
\begin{theorem}\label{thm:generalattack}
For at least one attacker $E$ participating in a general attack, the differential entropy of $R$ 
given side information held by $E$ follows the inequality
\begin{equation}
    h(R|E)\geq\frac{1}{2}\log \frac{4\pi}{1+\sigma^{-2}},
\label{attackerEntropy}
\end{equation}
where $\sigma$ is the same as defined in Section~\ref{sec:protocolEPR}. 
Furthermore, 
this attacker's response $r'$ satisfies the inequality 
\begin{equation}
    {\mathbb E}(R-r')^2\geq \frac{2}{e}\cdot\frac{1}{1+\sigma^{-2}}.
\label{VarGeneralAttack}
\end{equation}
\end{theorem}

\begin{proof}
In the entanglement-based protocol, the verifiers perform a heterodyne measurement. 
This is achieved by mixing one half of the TMS state with vacuum (denoted by $O$) and then performing a 
homodyne measurement per mode, in orthogonal directions $\theta$ and $\theta+\frac\pi2$, so
\begin{equation}
    \rho_{VP} \stackrel[\text{with }O]{\mathrm{Mixing}}{\longrightarrow} \rho_{\bar{V}\bar{O}P},
\end{equation}
where the bar represents the modes after mixing. 
Here $P$ is the subsystem sent to the prover and $\bar V$ is the subsystem on which the $\theta$ measurement will be applied.

The attackers (Alice and Bob) perform a quantum operation on the mode $P$ and any ancilla mode. We call the subsystem that Alice holds as $A$, and the one sent to Bob as $B$. 
The resulting state is $\rho_{\bar V\bar O AB}$.
We are interested in the tripartite state $\rho_{\bar V AB}$. We write the result of a homodyne measurement on $\bar V$ under angle $\theta$ as $U_\theta\in\mathbb R$,
and we write $\bar\theta=\theta+\frac\pi2$.
Lemma~\ref{lemma:uncertainty} gives
\begin{equation}
    \forall {\theta\in\cal A}\quad\quad h(U_\theta|A) + h(U_{\bar \theta}|B) \geq \log 2\pi.
\end{equation}
Averaging over $\theta$, and using the fact that averaging over $\bar\theta$ is the same
as averaging over $\theta$, gives
\begin{align}
    {\mathbb E}_{\theta\in\cal A}h(U_\theta|A) + {\mathbb E}_{\theta\in\cal A}h(U_{\bar \theta}|B) &\geq \log 2\pi
    \\
    \implies
    {\mathbb E}_{\theta\in\cal A}h(U_\theta|A) + {\mathbb E}_{\theta\in\cal A}h(U_{\theta}|B) &\geq \log 2\pi.
\end{align}
The last expression can be written as
\begin{equation}
    h(U|A\Theta) + h(U|B\Theta) \geq \log 2\pi,
\end{equation}
where the angle $\Theta$ is now represented as a random variable.
It follows that
\begin{equation}
    \max \Big\{h(U|A\Theta) ,\; h(U|B\Theta) \Big\}\geq \frac{1}{2}\log 2\pi.
\end{equation}
Finally, we note that $R= U\sqrt2 \tanh\zeta$ (with $\sinh\zeta=\sigma$)
and use Lemma~\ref{lemma:hscaling} to conclude
\begin{equation}
   h(R|E)\geq \frac{1}{2}\log 2\pi+ \frac{1}{2}\log \frac2{1+\sigma^{-2}}
   =\frac{1}{2}\log\frac{4\pi}{1+\sigma^{-2}},
\end{equation}
Here, we have set $\max \Big\{h(R|A\Theta) ,\; h(R|B\Theta) \Big\} = h(R|E) $. 
The result for ${\mathbb E}(R-r')^2$
follows directly from the Fano inequality
(Theorem \ref{theorem:fano}).
\end{proof}

\subsection{Comparison between attacker and honest prover}

We will now work in the $\sigma \gg 1$ limit. 
The protocol works only if the attackers have more ignorance about the value $R$ than
the honest prover. Note that we assume that the attackers are powerful and have access to an ideal channel ($t=1, u=0$).
For $\sigma\to\infty$, the difference between their entropies (\ref{attackerEntropy}), (\ref{entropy prover}) satisfies
\begin{equation}\label{equ:hgap}
    h(R|E)-h(R|R') \geq \frac12 \log\left(\frac 4e\cdot \frac t{1+2u}\right).
\end{equation}
The argument of the logarithm needs to be larger than~$1$.
This is the case when
\begin{equation}
    t > \frac{e}{4} \approx 0.680
    \quad \wedge \quad
    u \leq \frac{t\cdot 4/e-1}{2}.
\end{equation}
Note that Fano's inequality applied to the honest prover's entropy
(\ref{entropy prover}) would yield the expression
$\mathbb{E}(\sqrt{t}R-r')^2\approx\Sigma^2$ (as $\sqrt{t}R|R'$ is Gaussian with mean $r'$ in large $\sigma$ limit), 
with $\Sigma^2$ as defined in (\ref{proverposterior}), evaluating to $\Sigma^2 \approx (1/2+u)/t$.
On the other hand, the expected error of the attacker is lower bound by $2/e \approx 0.74$, which is strictly greater than $(1/2+u)/t$ for certain parameter ranges, as depicted in Figure~\ref{fig:SecInsecPlot}. This proves security of the protocol against a general attack in these parameter ranges.

\begin{figure}[h]
    \centering
    \includegraphics[width=0.4\linewidth]{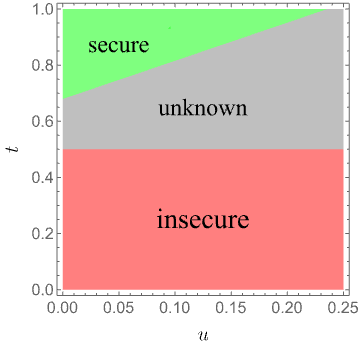}
    \caption{Security of the proposed CV-QPV protocol. For $t \leq 1/2$ it is insecure (red), as shown in \cite{qi2015loss}. For values in the green region, we prove security. Currently, no conclusions can be drawn about the grey region.}
    \label{fig:SecInsecPlot}
\end{figure}

\noindent As long as equation \eqref{equ:hgap} is positive, i.e.
\begin{align}
    \frac{t}{1+2u} > \frac{e}{4},
\end{align}
there's a finite gap between the attacker and the honest entropy about $R$. Then an attack fails if the score is greater than $\gamma$ (cf. Section~\ref{sec:protocol_pm}). To estimate the number of (independent) rounds $n$ we have to run for the attack success probability to become vanishingly small, we cannot assume a specific attack distribution and we have to assume the attackers have access to an ideal channel. We know that
\begin{align}
    \mathbb{E}(R-r')^2 \geq \frac{2}{e}, 
\end{align}
thus $\mathbb{E}(\sqrt{t}R-r')^2 \geq 2/e$ for any transmission $t$. The probability that the attackers' score falls below the threshold $\gamma$ is at most the probability that the score differs from $\mathbb{E}(\sqrt{t}R-r')^2/(1/2+u)$ by more than the difference $\Delta \overset{\mathrm{def}}{=}(2/e)/(1/2+u) - \gamma$\footnote{In the regime where $\Delta > 0$, which is the case for $u \lesssim 2/e - 1/2 \approx 0.24$ for sufficiently large $n$ where $\gamma \approx 1$. This value can also be observed in Figure \ref{fig:SecInsecPlot} at $t=1$.}. Thus we can use the Chebyshev inequality for the random variable of the score to get
\begin{align}
    \mathbb{P}\left[ \left| \frac{1}{n} \sum_{i=1}^n \frac{(\sqrt{t}R_i - r_i')^2}{1/2+u} - \frac{\mathbb{E} (\sqrt{t}R-r')^2}{1/2+u} \right| \geq \Delta \right] \leq \frac{\Tilde{\sigma}^2}{n\Delta^2} = O \left( \frac{1}{n \Delta^2} \right), 
\end{align}
where $\Tilde{\sigma}^2 = \mathbb{V} \left[ \frac{\left(\sqrt{t}R - r' \right)^2}{1/2+u} \right]$. If we set $n \Delta^2 = \Omega \left( \frac{1}{\varepsilon_{\mathrm{att}}} \right)$ then we get
\begin{align}
    \mathbb{P}\left[ \frac{1}{n} \sum_{i=1}^n \frac{(\sqrt{t}R_i - r_i')^2}{1/2+u} \leq \gamma \right] \leq O(\varepsilon_{\mathrm{att}}).
\end{align}

\section{Perfect attack with a single EPR pair} \label{sec:EPRattack}
It turns out that our protocol can be attacked if Alice and Bob pre-share one CV EPR pair (see Section~\ref{sec:prelims} for formal descriptions of CV entanglement and teleportation). The entanglement attack proceeds as follows:

\begin{enumerate}
    \item Alice and Bob pre-share an ideal EPR pair.
    \item 
    Alice teleports $\ket\psi$ to Bob.
    She forwards the measured displacement $(d_x,d_p)$ to Bob. 
    \item 
    Bob intercepts $\theta$ and immediately performs a homodyne measurement under angle $\theta$ on his own half of the EPR pair,
    obtaining outcome $\mu\in\mathbb R$.
    He forwards $\theta,\mu$ to Alice.
    \item
    Alice receives $\theta,\mu$. 
    She computes $r'=\mu-d_x\cos\theta-d_p\sin\theta$ and sends $r'$ to $V_1$.
    \item 
    Bob receives $d_x,d_p$. 
    He computes $r'=\mu-d_x\cos\theta-d_p\sin\theta$ and sends $r'$ to $V_2$.
\end{enumerate}
The state $\ket\psi$ is a coherent state with displacement $(x_0,p_0)$.
The effect of the teleportation is that Bob's half of the EPR pair
becomes a coherent state with displacement $(x_0+d_x,p_0+d_p)$.
Bob's homodyne measurement commutes with the teleport-induced displacement:
the undoing of the displacement can be done {\em after} Bob's measurement.
The noise in $r'$ with respect to $r$ is just shot noise, 
exactly as for the honest prover. Other noises originating from loss or excess noise can just be simulated by the attacker.

Hence, in the case of an ideal pre-shared EPR pair, the
responses from the attackers are statistically indistinguishable from
honest prover responses.

\section{Discussion}

The security analysis of CV-QPV differs from the discrete variable case, as the honest prover now responds with a sample from a probability distribution. 
Thus, to prove security (in the setting without pre-shared entanglement), 
we needed to show that an attack necessarily produces a different distribution than the honest one and that the verifiers can distinguish these distributions. 
We have shown that this can be done using an entropic uncertainty relation for the differential entropy together with a continuum version of the Fano inequality. We included attenuation and excess noise in the honest channel and showed security for a small range of parameters.
We further showed that the considered CV-QPV protocol is broken if one CV-EPR pair is pre-shared between the attackers. 

Since continuous-variable systems have some practical advantages over discrete ones (see Section~\ref{sec:Intro}) we hope that this work may spur interest into the further study of QPV in the context of continuous variables and we hope our techniques can be useful there. 

An immediate next step could be to extend this protocol to the case where the classical information $\theta$ is computed via a function $f(x,y)$ taking inputs $x, y$ from both verifiers, similar to the discrete variable \QPVBBf~protocol \cite{bluhm2022single,escolafarras2022singlequbit}, and to study CV entanglement attacks on that.

More generally, one may ask how far results on QPV for discrete variable protocols generalize or naturally carry over to the CV setting. 
For example, can the recent formulation of CV port-based teleportation \cite{pereira2023continuous} be used to immediately re-formalize the general attack on discrete variable QPV \cite{Beigi_2011} in the CV setting? 
Do the known attacks, which scale with properties of circuit decompositions of the provers' unitary \cite{speelman2016instantaneous, dolev2022non}, naturally generalize, for example to CV equivalents of $T$-count or $T$-depth? 

\subsubsection*{Acknowledgments}
We thank Kfir Dolev for interesting initial discussions on the topic of CV-QPV. RA was supported by the Dutch Research Council (NWO/OCW), as part of the Quantum Software Consortium programme (project number 024.003.037). PVL was supported by the Dutch Research Council (NWO/OCW), as part of the NWO Gravitation Programme Networks (project number 024.002.003). FS and LEF are supported by the Dutch Ministry of Economic Affairs and Climate Policy (EZK), as part of the Quantum Delta NL programme. B\v{S} and AAR acknowledge the support from Groeifonds Quantum Delta NL KAT2.


\bibliographystyle{alpha}
\bibliography{biblio}

\newcommand{\etalchar}[1]{$^{#1}$}
\begin{thebibliography}{GAW{\etalchar{+}}03}

\bibitem[ABM{\etalchar{+}}23]{allerstorfer2023relating}
Rene Allerstorfer, Harry Buhrman, Alex May, Florian Speelman, and Philip
  {Verduyn Lunel}.
\newblock Relating non-local quantum computation to information theoretic
  cryptography.
\newblock {\em arXiv preprint arXiv:2306.16462}, 2023.

\bibitem[ABSV21]{SWAP_protocol_Rene_et_all}
Rene Allerstorfer, Harry Buhrman, Florian Speelman, and Philip {Verduyn Lunel}.
\newblock Towards practical and error-robust quantum position verification.
\newblock {\em arXiv preprint arXiv:2106.12911}, 2021.

\bibitem[ABSV22]{allerstorfer2022role}
Rene Allerstorfer, Harry Buhrman, Florian Speelman, and Philip {Verduyn Lunel}.
\newblock On the role of quantum communication and loss in attacks on quantum
  position verification.
\newblock {\em arXiv preprint arXiv:2208.04341}, 2022.

\bibitem[ALS10]{cvqi_andersen_10}
Ulrik~L. Andersen, Gerd Leuchs, and Christine Silberhorn.
\newblock Continuous-variable quantum information processing.
\newblock {\em Laser \& Photonics Reviews}, 4(3):337--354, 2010.

\bibitem[BCF{\etalchar{+}}14]{Buhrman_2014}
Harry Buhrman, Nishanth Chandran, Serge Fehr, Ran Gelles, Vipul Goyal, Rafail
  Ostrovsky, and Christian Schaffner.
\newblock Position-based quantum cryptography: Impossibility and constructions.
\newblock {\em {SIAM} Journal on Computing}, 43(1):150--178, jan 2014.

\bibitem[BCS22]{bluhm2022single}
Andreas Bluhm, Matthias Christandl, and Florian Speelman.
\newblock A single-qubit position verification protocol that is secure against
  multi-qubit attacks.
\newblock {\em Nature Physics}, pages 1--4, 2022.

\bibitem[BFSS13]{Buhrman_2013}
Harry Buhrman, Serge Fehr, Christian Schaffner, and Florian Speelman.
\newblock The garden-hose model.
\newblock In {\em Proceedings of the 4th conference on Innovations in
  Theoretical Computer Science - {ITCS} {\textquotesingle}13}. {ACM} Press,
  2013.

\bibitem[BK11]{Beigi_2011}
Salman Beigi and Robert König.
\newblock Simplified instantaneous non-local quantum computation with
  applications to position-based cryptography.
\newblock {\em New Journal of Physics}, 13(9):093036, sep 2011.

\bibitem[BvL05]{cvqi_Braunstein_05}
Samuel~L. Braunstein and Peter van Loock.
\newblock Quantum information with continuous variables.
\newblock {\em Rev. Mod. Phys.}, 77:513--577, Jun 2005.

\bibitem[CGMO09]{OriginalPositionBasedCryptChandran2009}
Nishanth Chandran, Vipul Goyal, Ryan Moriarty, and Rafail Ostrovsky.
\newblock Position based cryptography.
\newblock In {\em Advances in Cryptology - CRYPTO 2009, 29th Annual
  International Cryptology Conference}, volume 5677 of {\em Lecture Notes in
  Computer Science}, pages 391--407. Springer, 2009.

\bibitem[CL15]{Chakraborty_2015}
Kaushik Chakraborty and Anthony Leverrier.
\newblock Practical position-based quantum cryptography.
\newblock {\em Physical Review A}, 92(5), nov 2015.

\bibitem[CLA01]{PhysRevA.63.052311}
Nicolas~J. Cerf, Mel L\'evy, and Gilles~Van Assche.
\newblock Quantum distribution of {Gaussian} keys using squeezed states.
\newblock {\em Phys. Rev. A}, 63:052311, Apr 2001.

\bibitem[CLP07]{book_cerf_cvqi}
Nicolas~J. Cerf, Gerd Leuchs, and Eugene~S Polzik.
\newblock {\em Quantum information with continuous variables of atoms and
  light}.
\newblock World Scientific, 2007.

\bibitem[CM22]{cree2022code}
Sam Cree and Alex May.
\newblock Code-routing: A new attack on position-verification.
\newblock {\em arXiv preprint arXiv:2202.07812}, 2022.

\bibitem[Cov99]{cover1999elements}
Thomas~M. Cover.
\newblock {\em Elements of information theory}.
\newblock John Wiley \& Sons, 1999.

\bibitem[DC22]{dolev2022non}
Kfir Dolev and Sam Cree.
\newblock Non-local computation of quantum circuits with small light cones.
\newblock {\em arXiv preprint arXiv:2203.10106}, 2022.

\bibitem[Dol19]{dolev2019constraining}
Kfir Dolev.
\newblock Constraining the doability of relativistic quantum tasks.
\newblock {\em arXiv preprint arXiv:1909.05403}, 2019.

\bibitem[EFS22]{escolafarras2022singlequbit}
Llorenç Escolà-Farràs and Florian Speelman.
\newblock Single-qubit loss-tolerant quantum position verification protocol
  secure against entangled attackers, 2022.

\bibitem[FBT{\etalchar{+}}14]{furrer2014position}
Fabian Furrer, Mario Berta, Marco Tomamichel, Volkher~B. Scholz, and Matthias
  Christandl.
\newblock Position-momentum uncertainty relations in the presence of quantum
  memory.
\newblock {\em Journal of Mathematical Physics}, 55(12), 2014.

\bibitem[GAW{\etalchar{+}}03]{article_Grossmman_03}
Frédéric Grosshans, Gilles Assche, Jerome Wenger, Rosa Brouri, Nicolas~J.
  Cerf, and Philippe Grangier.
\newblock Quantum key distribution using {Gaussian}-modulated coherent states.
\newblock {\em Nature}, 421:238--41, 02 2003.

\bibitem[GC19]{gonzales2019bounds}
Alvin Gonzales and Eric Chitambar.
\newblock Bounds on instantaneous nonlocal quantum computation.
\newblock {\em IEEE Transactions on Information Theory}, 66(5):2951--2963,
  2019.

\bibitem[GCW{\etalchar{+}}03]{10.5555/2011564.2011570}
Fr\'{e}d\'{e}ric Grosshans, Nicolas~J. Cerf, J\'{e}r\^{o}me Wenger, Rosa
  Tualle-Brouri, and Philippe Grangier.
\newblock Virtual entanglement and reconciliation protocols for quantum
  cryptography with continuous variables.
\newblock {\em Quantum Info. Comput.}, 3(7):535–552, oct 2003.

\bibitem[GG02]{PhysRevLett.88.057902}
Fr\'ed\'eric Grosshans and Philippe Grangier.
\newblock Continuous variable quantum cryptography using coherent states.
\newblock {\em Phys. Rev. Lett.}, 88:057902, Jan 2002.

\bibitem[GLW16]{gao2016quantum}
Fei Gao, Bin Liu, and QiaoYan Wen.
\newblock Quantum position verification in bounded-attack-frequency model.
\newblock {\em SCIENCE CHINA Physics, Mechanics \& Astronomy}, 59(11):1--11,
  2016.

\bibitem[GPS07]{garcia2007quantum}
Raul Garcia-Patron~Sanchez.
\newblock {\em Quantum information with optical continuous variables: from Bell
  tests to key distribution}.
\newblock PhD thesis, Universit{\'e} libre de Bruxelles, 2007.

\bibitem[Hil00]{PhysRevA.61.022309}
Mark Hillery.
\newblock Quantum cryptography with squeezed states.
\newblock {\em Phys. Rev. A}, 61:022309, Jan 2000.

\bibitem[HSH99]{PhysRevA.59.1820}
Alexander~S. Holevo, Masaki Sohma, and Osamu Hirota.
\newblock Capacity of quantum {Gaussian} channels.
\newblock {\em Phys. Rev. A}, 59:1820--1828, Mar 1999.

\bibitem[KMSB06]{PatentKentANdOthers}
Adrian Kent, William Munro, Timothy Spiller, and Raymond Beausoleil.
\newblock Tagging systems. {US} patent nr.\ 2006/0022832, 2006.

\bibitem[Lev09]{leverrier:tel-00451021}
Anthony Leverrier.
\newblock {\em {Theoretical study of continuous-variable quantum key
  distribution}}.
\newblock Theses, {T{\'e}l{\'e}com ParisTech}, November 2009.

\bibitem[LL11]{Lau_2011}
Hoi-Kwan Lau and Hoi-Kwong Lo.
\newblock Insecurity of position-based quantum-cryptography protocols against
  entanglement attacks.
\newblock {\em Physical Review A}, 83(1), jan 2011.

\bibitem[LLQ22]{liu_et_al:LIPIcs.ITCS.2022.100}
Jiahui Liu, Qipeng Liu, and Luowen Qian.
\newblock {Beating classical impossibility of position verification}.
\newblock In {\em 13th Innovations in Theoretical Computer Science Conference
  (ITCS 2022)}, volume 215 of {\em Leibniz International Proceedings in
  Informatics (LIPIcs)}, pages 100:1--100:11, Dagstuhl, Germany, 2022. Schloss
  Dagstuhl -- Leibniz-Zentrum f{\"u}r Informatik.

\bibitem[LM00]{laurent2000adaptive}
Beatrice Laurent and Pascal Massart.
\newblock Adaptive estimation of a quadratic functional by model selection.
\newblock {\em Annals of statistics}, pages 1302--1338, 2000.

\bibitem[LPF{\etalchar{+}}18]{laudenbach2018continuous}
Fabian Laudenbach, Christoph Pacher, Chi-Hang~Fred Fung, Andreas Poppe,
  Momtchil Peev, Bernhard Schrenk, Michael Hentschel, Philip Walther, and
  Hannes H{\"u}bel.
\newblock Continuous-variable quantum key distribution with {Gaussian}
  modulation—the theory of practical implementations.
\newblock {\em Advanced Quantum Technologies}, 1(1):1800011, 2018.

\bibitem[Mal10a]{Malaney_2010_b}
Robert~A. Malaney.
\newblock Location-dependent communications using quantum entanglement.
\newblock {\em Physical Review A}, 81(4), apr 2010.

\bibitem[Mal10b]{Malaney_2010_a}
Robert~A. Malaney.
\newblock Quantum location verification in noisy channels.
\newblock In {\em 2010 {IEEE} Global Telecommunications Conference {GLOBECOM}
  2010}. {IEEE}, dec 2010.

\bibitem[PBP23]{pereira2023continuous}
Jason~L. Pereira, Leonardo Banchi, and Stefano Pirandola.
\newblock Continuous variable port-based teleportation.
\newblock {\em arXiv preprint arXiv:2302.08522}, 2023.

\bibitem[QS15]{qi2015loss}
Bing Qi and George Siopsis.
\newblock Loss-tolerant position-based quantum cryptography.
\newblock {\em Physical Review A}, 91(4):042337, 2015.

\bibitem[Ral99]{PhysRevA.61.010303}
Timothy~C. Ralph.
\newblock Continuous variable quantum cryptography.
\newblock {\em Phys. Rev. A}, 61:010303, Dec 1999.

\bibitem[Rei00]{PhysRevA.62.062308}
Margaret~D. Reid.
\newblock Quantum cryptography with a predetermined key using
  continuous-variable {Einstein-Podolsky-Rosen} correlations.
\newblock {\em Phys. Rev. A}, 62:062308, Nov 2000.

\bibitem[Spe16]{speelman2016instantaneous}
Florian Speelman.
\newblock Instantaneous non-local computation of low {T}-depth quantum
  circuits.
\newblock In {\em 11th Conference on the Theory of Quantum Computation,
  Communication and Cryptography (TQC 2016)}. Schloss Dagstuhl-Leibniz-Zentrum
  fuer Informatik, 2016.

\bibitem[Unr14]{Unruh_2014_QPV_random_oracle}
Dominique Unruh.
\newblock Quantum position verification in the random oracle model.
\newblock In {\em Advances in Cryptology -- CRYPTO 2014}, pages 1--18, Berlin,
  Heidelberg, 2014. Springer Berlin Heidelberg.

\bibitem[Vai94]{vaidman1994teleportation}
Lev Vaidman.
\newblock Teleportation of quantum states.
\newblock {\em Physical Review A}, 49(2):1473, 1994.

\bibitem[Wig32]{wigner1932quantum}
Eugene Wigner.
\newblock On the quantum correction for thermodynamic equilibrium.
\newblock {\em Physical review}, 40(5):749, 1932.

\bibitem[WLB{\etalchar{+}}04]{PhysRevLett.93.170504}
Christian Weedbrook, Andrew~M. Lance, Warwick~P. Bowen, Thomas Symul,
  Timothy~C. Ralph, and Ping~Koy Lam.
\newblock Quantum cryptography without switching.
\newblock {\em Phys. Rev. Lett.}, 93:170504, Oct 2004.

\bibitem[WZ82]{Wootters1982NoCloning}
William~K. Wootters and Wojciech Zurek.
\newblock A single quantum cannot be cloned.
\newblock {\em Nature}, 299:802--803, 1982.

\end{thebibliography}

\end{document}